\documentclass[a4paper,twocolumn]{article}
\usepackage{vmargin}
\setmarginsrb{0.7in}{1in}{0.7in}{1in}{0mm}{0mm}{5mm}{5mm}

\usepackage{times}

\usepackage{soul}
\usepackage{url}
\usepackage[hidelinks]{hyperref} 
\usepackage[utf8]{inputenc}
\usepackage[small]{caption}
\usepackage{graphicx}
\usepackage{amsmath, amsfonts, amsthm,amssymb}
\usepackage{booktabs}
\usepackage{makecell}
\usepackage{graphicx}
\graphicspath{{figures/}}
\usepackage{xspace}
\usepackage{subfigure}
\usepackage{multirow}
\usepackage{booktabs} 

\usepackage[ruled,vlined,linesnumbered,algo2e]{algorithm2e}
\SetKwFor{Function}{function}{is}{end function}
\SetKw{Input}{input}
\SetKw{Output}{output}
\SetKwInOut{Input}{input}
\SetKwInOut{Output}{output}

\newtheorem{theorem}{Theorem}

\newcommand{\game}{flowNCG\xspace}
\newcommand{\avggame}{avg-flowNCG\xspace}
\newcommand{\mingame}{min-flowNCG\xspace}
\newcommand{\stp}{\mathbf{s}} 
\newcommand{\opt}{OPT\xspace}
\newcommand{\bA}{well-connected\xspace}

\usepackage{todonotes}

\title{Flow-Based Network Creation Games}

\author{
Hagen Echzell\and
Tobias Friedrich\and
Pascal Lenzner\and
Anna Melnichenko
}
\date{Hasso Plattner Institute, University of Potsdam, Germany}

\begin{document}
\maketitle
\begin{abstract}
Network Creation Games~(NCGs) model the creation of decentralized communication networks like the Internet. In such games strategic agents corresponding to network nodes selfishly decide with whom to connect to optimize some objective function. Past research intensively analyzed models where the agents strive for a central position in the network. This models agents optimizing the network for low-latency applications like VoIP. However, with today's abundance of streaming services it is important to ensure that the created network can satisfy the increased bandwidth demand. To the best of our knowledge, this natural problem of the decentralized strategic creation of networks with sufficient bandwidth has not yet been studied.  

We introduce Flow-Based NCGs where the selfish agents focus on bandwidth instead of latency. In essence, budget-constrained agents create network links to maximize their minimum or average network flow value to all other network nodes. Equivalently, this can also be understood as agents who create links to increase their connectivity and thus also the robustness of the network.
For this novel type of NCG we prove that pure Nash equilibria exist, we give a simple algorithm for computing optimal networks, we show that the Price of Stability is~1 and we prove an (almost) tight bound of $2$ on the Price of Anarchy. Last but not least, we show that our models do not admit a potential function.     
\end{abstract}

\section{Introduction}
Many of the networks we crucially rely on nowadays have evolved from small centrally designed networks into huge networks created and decentrally controlled by many selfish agents with possibly conflicting goals. For example, the structure of the Internet is essentially the outcome of a repeated strategic interaction by many selfish economic agents~\cite{Pap01,T04}, e.g., Internet service providers (ISPs). Hence, such Internet-like networks can be analyzed by considering a strategic game which models the interaction of the involved agents. This insight has inspired researchers to propose game-theoretic models for the decentralized formation of networks, most prominently the models in \cite{Myerson13,JW96,BG00,fabrikant2003network} and many variants of them.
In these models selfish agents select strategies to maximize their utility in the created network. Thus, the structure of the created network and corresponding agent utilities depend on the selfishly chosen strategies of all agents.

To the best of our knowledge, the utility of the involved agents in all the existing models depends on the agents' centrality or the size of their connected component and the cost spent for creating links. Striving for centrality, e.g., for low hop-distances to all other nodes in the created network, certainly plays an important role for the involved selfish agents, e.g., ISPs, but it is not the only driving force. What has been completely neglected so far are bandwidth considerations. 

In this paper we take the first steps into this largely unexplored area of game-theoretic network formation with bandwidth maximization. In particular, we consider budget-constrained agents that strategically form links to other agents to maximize their flow value towards all other nodes. This models agents that optimize the networks for data-intensive applications like online streaming instead of low-latency applications like VoIP. Interestingly, by the Max Flow Min Cut Theorem~\cite{ford_fulkerson_1956}, flow maximization corresponds to connectivity maximization which is tightly related to network robustness. While there have been works which consider robustness aspects as side constraints, also the aspect of maximizing the created network's robustness is, to the best of our knowledge, entirely unexplored.

While classical network formation models with focus on centrality yield equilibrium networks that are sparse and have low diameter, our work shows that focusing on bandwidth/connectivity may explain why densely connected sub-networks, like $k$-core structures, and larger cycles appear in real-world networks: they are essential for its connectivity.

\subsection{Model and Notation}
We propose the \textit{Flow-Based Network Creation Game (\game)}, which is a variant of the well-known Network Creation Game~\cite{fabrikant2003network} where the agents create edges to maximize their flow value as defined by the classical Max-Flow Problem~\cite{ford_fulkerson_1956}. Given a set of $n$ selfish agents which corresponds to the set of nodes $V$ of a weighted directed network $G=(V,E)$. Each agent has a uniform budget of $k\in\mathbb{N}$, with $k<n$, to establish connections to other agents. In the created graph $G$ each edge has a capacity $c:E(G)\rightarrow \mathbb{N}$. We define the \emph{degree} of a node $v$ in $G$ as $deg_G(v) = \sum_{(x,v)\in E,x\in V}c(x,v) + \sum_{(v,y)\in E,y\in V}c(v,y)$, i.e., as the sum of the capacities of all incoming and outgoing edges.

The edge set and capacities are defined by the strategies of the agents. In particular, each agent strategically decides how to spend its budget, i.e., which subset of incident edges to buy and for each bought edge its capacity. Therefore, the strategy $S_v$ of an agent $v$ is a set of tuples $(x,c(v,x))\in V\setminus\{v\}\times [k]$, where $[k] = \{1,\dots,k\}$. Here $x$ denotes the node to which an edge is formed and $c(v,x)$ the capacity of that edge bought by agent~$v$. For nodes $w\in V$ to which no edge is formed, we assume that $c(v,w) = 0$. A strategy $S_v$ is feasible if the total capacity of all edges in $S_v$ does not exceed the given budget $k$, i.e., $\sum_{x\in V\setminus\{v\}}{c(v,x)}\leq k$. If $(x,c(v,x))\in S_v$, we call $v$ the \textit{owner} of the edge $(v,x)$ with capacity $c(v,x)$. The vector of all agents' strategies $\mathbf{s}=(S_{v_1},\ldots, S_{v_n})$ is denoted as a \textit{strategy profile}. Any strategy profile $\stp$ uniquely specifies the weighted directed network $G(\stp)=(V, E(\stp))$ where $(v,x)\in E(\stp)$ of capacity $c(v,x)$ if and only if agent $v$ buys the edge to $x$ with non-zero capacity $c(v,x)$. 

For sending flow in the network $G(\stp)$, we assume its undirected version $F(\stp) = (V,E_F(\stp))$ where $\{v,x\} \in E_F(\stp)$ if $(v,x) \in E(\stp)$ or $(x,v) \in E(\stp)$ and the capacity $c(\{v,x\})$ is defined as $c(\{v,x\}) = c(v,x) + c(x,v)$. Thus, for sending flow, any edge of $G(\stp)$ can be 
used in both directions. See Figure~\ref{fig:graph_instance} for an illustration of $G(\stp)$ and its corresponding undirected version $F(\stp)$.  
\begin{figure}[h]

\center{\includegraphics[width=7cm]{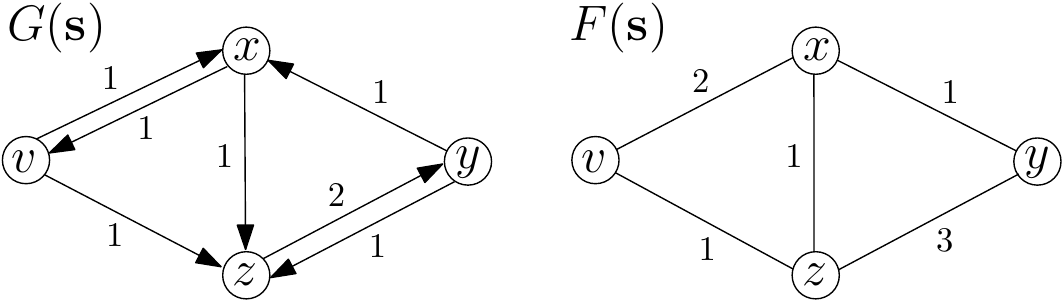}}

\caption{Graph $G(\stp)$ and its undirected version $F(\stp)$ for $\stp = (S_v,S_x,S_y,S_z)$ with $S_v = \{(x,1),(z,1)\}$, $S_x = \{(v,1),(z,1)\}$, $S_y = \{(x,1),(z,1)\}$, $S_z=\{(y,2)\}$. Agent $z$ has degree $5$ and can send a flow of value 3, 3 and 4 to the agents $v, x$ and $y$ respectively.}
\label{fig:graph_instance}
\end{figure} 
Note that throughout the paper we only work with $G(\stp)$ and we will always assume its undirected version $F(\stp)$ when computing flow values.

We call a directed weighted network $G$ \textit{feasible} if the sum over the weights of all incident outgoing edges for each node is at most $k$. Thus, we have a bijection between the set of feasible networks and the set of strategy profiles. Therefore, we will further omit $\stp$ in $G(\stp)$ if it is clear from the context and we will use $G$ and $\stp$ interchangeably if $G = G(\stp)$.

Let $\lambda_{G(\stp)}(v,x)=\sum_{v_1\in V_1, v_2\in V_2, \{v_1,v_2\}\in E_F(\stp)}{c(\{v_1,v_2\})}$ be the capacity of a minimum $v$-$x$-cut $(V_1,V_2)$ disconnecting node $v$ from $x$ in $F(\stp)$ which we also call the \textit{local edge connectivity} between $x$ and $v$. The minimum over all pairwise local edge connectivities in $G(\stp)$ is called the \textit{edge connectivity} of $G(\stp)$ and is denoted as $\lambda(G(\stp))$. 
A cut $(V_1, V_2)$ in $G(\stp)$, where $V_1\cup V_2 = V$ and $V_1\cap V_2=\emptyset$, such that $\lambda(G(\stp))=\sum_{v_1\in V_1, v_2\in V_2, \{v_1,v_2\}\in E_F(\stp)}{c(\{v_1,v_2\})}$ is called a \textit{minimum cut (min-cut)} of $G(\stp)$. Note that cuts in $G(\stp)$ are always defined on its undirected version $F(\stp)$.
By the Max Flow Min Cut Theorem~\cite{ford_fulkerson_1956}, the size of the min-$v$-$x$-cut equals the value of the max-flow between nodes $v$ and $x$. Hence, we use the local edge connectivity of $v$ and $x$ interchangeably with the value of maximum $v$-$x$-flow. 

We define two versions of the \game. The first is called the \textit{Average Flow NCG (\avggame)}. For a given strategy profile $\stp$, the utility function of an agent $v$ is defined as the average maximum flow value to each other node in the network: $u(v, \stp)=\sum_{i\in V\setminus \{v\}}{\frac{\lambda_{G(\stp)}(v,i)}{n-1}}$. The \textit{social utility} of network $G(\stp)$ is the average agent utility, i.e., $u(\stp) = \sum_{v\in V}\frac{u(v, \stp)}{n}.$ 

The second variant is called the \textit{Minimum Flow NCG (\mingame)}. The utility of an agent in network $G(\stp)$ is the size of the min-cut of the network $\lambda(G(\stp))$. 
Note that this value is the same for all agents in $G(\stp)$. 
We define a tie-breaking between strategies for an agent which yield the same min-cut of $G(\stp)$: agents try to maximize the number of nodes with local edge connectivity exceeding $\lambda(G(\stp))$, i.e., to increase the number of highly robust connections. In particular, we define the agent's utility as the following vector: $u(v,\stp) =\left(u_1(v,\stp), u_2(v, \stp)\right)$, where $u_1(v,\stp) = \lambda(G(\stp))$ and $u_2(v,\stp) =  \vert\{i \in V\!\setminus\!\{v\}:\ \lambda_{G(\stp)}(i,v)>\lambda(G(\stp))\}\vert$.
We call the second component of the utility function $u_2(v,\stp)$ as the number of \emph{\bA nodes}. We assume that agents aim for maximizing their utility vector lexicographically.
We define the social utility of a network in the \mingame as the edge connectivity of the network: $u(\stp)=\lambda(G(\stp)).$
The \textit{social optimum (\opt)} is  a network $G(s^*)$ which maximizes the social utility $u(s^*)$ over all feasible strategy profiles.

We say agent $v$ \textit{reduces the capacity} of the $v$-$x$ edge by $\ell$ if it either decreases the capacity $c(v,x)$ by $\ell$ (if $c(v,x)>\ell$) or if it deletes the edge $(v,x)$ (if $c(v,x)=\ell$). Analogously, agent $v$ \textit{increases the capacity} of the $v$-$x$ edge if it either increases the capacity $c(v,x)$ by $\ell$ (if $(x, c(v,x))\in S_v$) or if it buys the edge $(v,x)$ with capacity $\ell$.  

An \textit{improving move} for agent $v_i$ is a strategy change from $S_{v_i}$ to $S'_{v_i}$ such that $u(v_i, (s'_{v_i}, s_{-v_i})) > u(v_i, s)$, where $(s'_{v_i}, s_{-v_i})$ is a new strategy profile which only differs from $s$ in the strategy of agent $v_i$. We say that agent $v$ plays its \textit{best response} $S_v$ if there is no improving move for agent $v$. A strategy change towards a best response is called \textit{best response move}. A sequence of best response moves which starts and ends with the same network is called a \textit{best response cycle}. If every sequence of improving moves is finite, then the game has the \textit{finite improving property (FIP)} or, equivalently, the game is a potential game~\cite{monderer1996potential}.

We say that a network $G(\stp)$ is in \textit{pure Nash equilibrium (NE)} if all agents play a best response in $G(\stp)$. We measure the loss of social utility due to the lack of a central authority and the agents’ selfishness with the \textit{Price of Anarchy (PoA)}~\cite{koutsoupias1999worst}. Let $minNE(n,k)$($maxNE(n,k)$) be the minimum (maximum) social utility of a NE with $n$ agents with budget $k$, and let $opt(n,k)$ be the corresponding utility of the social optimum. For given $n$ and $k$ the PoA is $\frac{opt(n,k)}{minNE(n,k)}$ and the \textit{Price of Stability (PoS)} is $\frac{opt(n,k)}{maxNE(n,k)}$. The latter measures the minimum sacrifice in social utility for obtaining a NE network.

\subsection{Related Work}
The NCG was proposed in~\cite{fabrikant2003network} and has been an object of intensive study for almost two decades. The model depends on an edge-price parameter $\alpha$ which defines the cost of any edge of the network and the agents' objective function is to minimize their closeness centrality, i.e., their average hop-distance to all other nodes in the created network. A constant PoA was shown for almost all values of $\alpha$, except for the range where $\alpha \approx n$, where only an upper bound of $o(n^\varepsilon)$, for any $\varepsilon > 0$, is known~\cite{De07}. See~\cite{alvarez2019price,BiloL20} for the latest improvements on the PoA and a more detailed discussion. Moreover,~\cite{KL13} proved that the NCG does not have the FIP. 

Many variants of the NCG have been proposed and analyzed but so far all of them consider agents which strive for centrality or only for creating a connected network. Related to our model are versions where agents can only swap edges~\cite{ADHL13,MS12}, bounded-budget versions~\cite{laoutaris2008bounded,ehsani2015bounded} and two versions which focus on robustness: in~\cite{MMO15} agents maintain two vertex-disjoint paths to any other node and in~\cite{CLMM16} agents consider their expected centrality with respect to a single random edge failure. Thus agents in these models enforce 2-vertex-connectivity or 2-edge-connectivity in the created network which are much weaker robustness concepts compared to our approach. Another line of research are Network Formation Games~\cite{BG00} where agents only strive for being connected to all other nodes. Several variants with focus on robustness have been introduced: In~\cite{BG00_reliability} a model with probabilistic edge failures, in~\cite{Kli11,kliemann2017swap} adversarial models where a single edge is removed after the network is formed,  in~\cite{Goyal16,FriedrichIKLNS17} a version where a node is attacked with deterministic spread to neighboring nodes, and in~\cite{chen2019network} a model with probabilistic spread. To the best of our knowledge, no related model exists where agents try to maximize their connectivity.  

In the realm of classical optimization, maximizing the robustness of a network with a given budget is a frequently studied network augmentation problem, see, e.g.,~\cite{watanabe1987edge,frank1994connectivity,nagamochi2008algorithmic} for surveys. However, due to the centralized optimization view, these problems are very different from our model.

 

\subsection{Our Contribution}

By incorporating bandwidth considerations into the classical NCG we shed light on a largely unexplored area of research. In our Flow-Based NCG agents strategically create links under a budget constraint to maximize their average or minimum flow value to all other agents. This is in stark contrast to existing models which focus on agents aiming for centrality or for maximizing the size of their connected component. 

For our novel modes, we provide an efficient algorithm to compute a social optimum network and we uncover important structural properties of it.
A major part of our research is a rigorous study of the structure and quality of the induced equilibrium networks. We show that NE networks are guaranteed to be connected and that they contains a subgraph that is at least $(k+1)$-edge-connected. Moreover, any NE in the \mingame is always $(k+1)$-edge-connected. Most importantly, we prove that the social utility of all NE networks is close to optimum. More precisely, our PoA results for both models guarantee that the social utility of any NE is at least half the optimal utility, i.e., PoA $\leq 2$, and this bound is tight for the \mingame and almost tight for the \avggame. 
Besides this, we prove for both versions that the PoS $=1$ and that the FIP does not hold.  
See Table~\ref{tab:contribution} for an overview. 

Due to limited space, some proofs are sketched or omitted.

\begin{table*}[h]
\begin{center}
\renewcommand{\arraystretch}{1.3}
\begin{tabular}{@{} l l l l @{}l l l @{}}
\toprule
	\textbf{Model} & \textbf{u(\opt)} & \textbf{PoS} & \textbf{PoA} && \textbf{FIP} \\
	\hline
	\textbf{\mingame} & $2k$ (Thm.~\ref{thm:opt_produced_by_the_alg}) & $1$ (Thm.~\ref{thm:PoS}) & $\frac{2k}{k+1}$ (Thm.~\ref{thm:PoA_mingame}) && no (Thm.~\ref{thm:IRC}) \\
	\cmidrule{1-6}
	\multirow{3}{*}{\textbf{\avggame}} & \multirow{3}{*}{$2k$ (Thm.~\ref{thm:opt_produced_by_the_alg})} & \multirow{3}{*}{$1$ (Thm.~\ref{thm:PoS})} & \makecell[l]{$< 2$ (Thm.~\ref{thm:PoA_UB_avggame}) } && \multirow{3}{*}{no (Thm.~\ref{thm:IRC})} \\
	&&&$\ge \frac{2k}{k+\frac{k(k-1)}{n-1}}$ & if $k\leq 0.5+\sqrt{n-0.75}$ (Thm.~\ref{thm:PoA_LB_avggame})\\
	&&& $\ge \frac{2k}{k+1}$  & if $k> 0.5+\sqrt{n-0.75}$ (Thm.~\ref{thm:PoA_LB_avggame})
\end{tabular}
\caption{Overview of our results}
\label{tab:contribution}
\end{center}
\end{table*}

\section{Social Optimum}
We analyze the structure of the \opt networks. In particular, we show a generic network construction with maximum social utility $2k$ for both game models.

\begin{algorithm2e}[h]
\SetKwInOut{Input}{input}
\SetKwInOut{Output}{output}
 \Input{ unweighted directed clique $K_n\!=\!(V,E)$}
 \Output{ social optimum network $G$}
 $G\leftarrow (V,\emptyset)$\; 
 \For{$i\leftarrow 1$ \KwTo $k$}{
 find a directed Hamiltonian cycle $C_n$ in $K_n$\; 
 add the cycle to the output graph $G\leftarrow G\cup C_n$, i.e., increase capacity of each $(u,v)$-connection in $G$ by 1 if $(u,v)\in C_n$\;
 }  	
	\caption{Algorithm for computing the \opt }\label{alg:opt_graphs}
\end{algorithm2e}

\begin{theorem}\label{thm:opt_produced_by_the_alg}
Algorithm~\ref{alg:opt_graphs} computes an optimal network with social utility $2k$ for the \avggame\ and \mingame in polynomial time.
\end{theorem} 
\begin{proof}
First, we prove a general upper bound on the social utility of \opt. Let $G$ be an optimal network in the \mingame. The social utility of $G$ equals the edge connectivity of $G$ which is at most $\min_{v\in V}{deg_G(v)}$. The handshake lemma yields $\sum_{v\in V}{deg_G(v)}=2\sum_{e\in E(G)}{c(e)}\leq 2nk$, where the last inequality holds because every agent can build edges of the total capacity of at most $k$. Since the minimum degree is at most the average degree in the network, we get $u(\stp)=\lambda(G(\stp))\leq \min_{v\in V}{deg_G(v)}\leq 2k$.

In the \avggame, the social utility is $u(\stp)=\frac{1}{n}\sum_{v\in V}{u(v,\stp)}$ where $u(v,\stp)$ is the average local edge connectivity over all pairs $(v,V\setminus\{v\})$. The local edge connectivity of any pair of nodes $v, w$ is at most $\min\{deg_G(v), deg_G(w)\}$. Thus, $u(\stp)\leq \frac{1}{n}\sum_{v\in V}deg(v)\leq 2k$, again by the handshake lemma.

Now we show that the social utility of the graph $G$ produced by the algorithm matches the upper bound which implies that $G$ is optimal. Since each node in $G$ has an outdegree of $k$ and no self-loops, $G$ is a feasible network for both games. For any two nodes $v, w\in V$ the minimum $v$-$w$-cut contains two edges of each Hamiltonian cycle added to $G$ during the algorithm's run. Thus, the local edge connectivity of any pair of nodes equals $2k$ which implies $u(G)=2k$ in both games.

It is easy to see that the algorithm runs in polynomial time since the Hamiltonian cycle can be constructed greedily.
\end{proof}

\begin{theorem}
In the \mingame and the \avggame OPT is a connected $2k$-regular network. 
\end{theorem}


\section{Structural Properties of \\Equilibria}
In this section we prove structural properties of NE networks, which we will use later on to derive our PoA bounds.

\begin{theorem}\label{thm:all_budget_use_minFNCG}
In the \mingame and the \avggame all agents use all of their k budget units in an NE network.
\end{theorem}
\begin{proof}
Consider a network $G=G(\stp)$ which is in NE. Assume towards a contradiction that there is an agent $v \in V$ which spends strictly less than $k$ units of its budget in $G (\stp)$.

For the \mingame, we have, by definition, $u(v,\stp)=(u_1(v,\stp), u_2(v,\stp))=(\lambda(G), u_2(v,\stp))$. Let $w\in V$ be an agent such that $\lambda_G(v,w)=\lambda(G)$. Then agent $v$ can improve its strategy by increasing the capacity of the edge $(v,w)$. Indeed, after the strategy change ($s$ to $s' = (s_v', s_{-v})$) either the edge connectivity of the network increases, i.e., $\lambda(G(s'))>\lambda(G(s))$ or the edge connectivity does not change but $u(v,s')=(u_1(v,s),u_2(v,s'))\geq(u_1(v,s),u_2(v,s)+1)>u(v,s)$. In both cases agent $v$ improves, which contradicts that $G$ is in NE. 

For the \avggame it is easy to see, that building an additional edge towards any agent $w \in V\setminus\{v\}$ increases their local edge connectivity, while not decreasing the local edge connectivity between any pair of agents, and thus this results in a strategy $s' = (s_v', s_{-v})$ such that $u_v(s') > u_v(s)$. Thus, agent $v$ can improve, which contradicts that $G$ is in NE.
\end{proof}

\begin{theorem}\label{thm:NE_is_connected_minFNCG}
In the \mingame and the \avggame any equilibrium network is connected.
\end{theorem}
\begin{proof}
For both games assume towards a contradiction that there is a disconnected equilibrium network $G$. Note that $\lambda(G)=0$. By Theorem~\ref{thm:all_budget_use_minFNCG} all agents use all of their budgets, thus, every vertex has an outgoing edge in $G$. As the number of nodes is finite, there must be a directed cycle in every weakly connected component of $G$. Thus, any agent $v$ in such a cycle can remove an existing outgoing cycle-edge without disconnecting the weakly connected component and add the edge to any node of another component in~$G$. 

In the \mingame, this move either increases the edge-connectivity of the network, i.e., it increases $u_1(v,G)$, or it increases the number of \bA agents, i.e., it increases $u_2(v,G)$ and $u_1(v,G)$ remains unchanged. 
Thus, agent $v$ can improve which contradicts that $G$ is in NE.

For the \avggame, consider two weakly connected components $X, Y\subseteq V$ in $G$ such that $|X|\leq |Y|$. Thus there must be an agent $v$ in $X$ which can reduce the capacity of one connection in $X$ by 1 and add an edge of capacity 1 to some other node in $Y$. This move changes only the local edge-connectivity between $v$ and any other node by at most $1$. Thus, its utility changes by at least $-\frac{|X|-1}{n-1}+\frac{|Y|}{n-1}>0$. Thus, agent $v$ can improve which contradicts that $G$ is in NE.
\end{proof}



\noindent The next property will be about \emph{$j$-clusters}, which are subset of nodes $C \subseteq V, |C| \geq 2$, from $G$ such that the $C$-induced subgraph of $G$ has edge connectivity of at least~$j$.

\begin{theorem}\label{thm:cluster}
For the \mingame and the \avggame, every NE network contains a $(k + 1)$-cluster.
\end{theorem}
\begin{proof}
Consider a NE network $G(\stp)$. We partition $G(\stp)$, or more precisely its undirected version $F(\stp)$, into components by the following method: As long as there is a cut $(X,Y)$ of size at most $k$, remove all edges of $(X,Y)$ from $G$. At each step of the algorithm either the algorithm terminates, i.e., there is a $(k+1)$-cluster in $G$, or the number of connected components increases by at least 1 while the number of edges decreases by at most $k$. By Theorem~\ref{thm:all_budget_use_minFNCG} the sum of edge weights in the initial graph is~$nk$. Hence, if a $(k+1)$-cluster was not found, the algorithm produces a network containing $n$ components and a set of edges of the total capacity of at least $k$. Since self-loops are not allowed, we have a contradiction. Thus, the process must find a $(k+1)$-cluster. 
\end{proof}



	
\section{Price of Stability}
\begin{theorem}\label{thm:PoS}
The PoS is 1 in the \mingame and the \avggame.
\end{theorem}
\begin{proof}
Consider a directed cycle $C$ where every edge has capacity $k$. We will show that $C$ is a NE in both games. 

For the \mingame, consider a node, say $v_1$, in $C=G(\stp)$. Agent $v_1$ owns one edge, w.l.o.g., $(v_1,v_2)$ of capacity $k$. The only strategy change that $v_1$ can perform is to reduce the capacity of the edge $(v_1,v_2)$ and to spend the rest of the budget on new edges to other nodes. This move strictly decreases the degree of the node $v_2$. Since the edge connectivity of any network is at most its minimum degree, which then is the degree of $v_2$, this implies that the strategy change decreases $v_1$'s utility. Hence, there is no improving move for agent $v_1$, which, by symmetry, implies that every agent plays best response in $C$ which proves that $C$ is in NE.

Analogously, in the \avggame the same strategy change by any of the agents, say $v_1$, decreases the degree of the respective neighbor $v_2$. Thus, the local edge connectivity of this pair of nodes decreases. At the same time, the local edge connectivity of $v_1$ with other nodes does not change since it does not exceed the minimum degree of the pair, i.e., the degree of node $v_1$. This implies a decrease of the average local edge connectivity for agent $v_1$, i.e., the decrease of $v_1$'s utility. Thus, all agents play best response and $C$ is a NE.
  
By Theorem~\ref{thm:opt_produced_by_the_alg} the social utility of \opt is $2k$. We have proved the existence of a NE network with the same social utility for both models. Thus, $PoS=1$ for both games. 
\end{proof}

\section{Price of Anarchy}
In this section we present our main results, which are a tight bound in the PoA for the \mingame and an almost tight PoA bound for the \avggame. Moreover, whereas all previous results show the similarity of the \mingame and the \avggame we will now highlight their differences. In particular, we show that NE networks in the \mingame may not be NE networks in the \avggame and vice versa.

\subsection{PoA of the Min-FlowNCG }
We start by providing a NE network for the \mingame with minimum social utility.
\begin{theorem}\label{thm:UB_NE_mingame}
For a given $n$ and $k$ such that $k<n$, there is a NE network  $G$ such that $u(G)=k+1$ in the \mingame.
\end{theorem}   
\begin{proof}[Proofsketch]
We consider the following generic construction for a given number of agents $n$ and budget $k<n$. Consider a set of nodes $\{v_1, \ldots, v_n\}=V(G)$ and a set of edges $E(G)$ where for all $1\leq i<k$ each node $v_i$ is connected with a node $v_{i+1}$ by an edge $(v_i,v_{i+1})$ with capacity $c(v_i, v_{i+1})=k-(i-1)$ and by an edge $(v_{i+1}, v_i)$ with capacity $c(v_{i+1}, v_i)=i$. Moreover, there is an edge $(v_k, v_n)$ with capacity 1 and a set of sequential edges, for all $ k\leq i\leq n-1, (v_{i+1}, v_i)$ with capacity $k$. See Figure~\ref{fig:NE_mingame} for illustration of the construction. We claim that $G$ is in NE in the \mingame. We omit the proof due to the space restrictions.
\end{proof}
\begin{figure}[ht]
\centering\includegraphics[width=.4\textwidth]{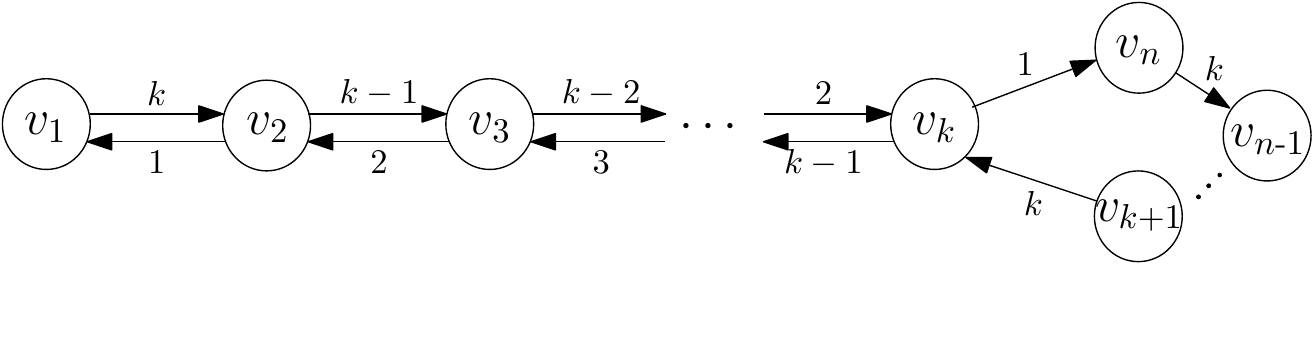}
\caption{NE network in the \mingame. Note that this network is not a NE in the \avggame. Indeed, agent $v_2$ can delete the edge $(v_2,v_1)$ and create the edge $(v_2,v_n)$ with capacity 1.}
\label{fig:NE_mingame}
\end{figure}

\noindent The next result shows that any NE network in the \mingame is $(k+1)$-edge-connected.

\begin{theorem}\label{thm:LB_NE_utility_mingame}
The social utility of any NE network is at least $k+1$ in the \mingame.
\end{theorem}
\begin{proof}
Assume towards a contradiction that there is a NE $G=(V,E)$ such that it contains a cut $(A,B)$ such that $\lambda(G)=\sum_{e\in (A,B)}c(e)\leq k$. By Theorem~\ref{thm:cluster} there is an induced subgraph $K=(C,E(K))\subseteq G$ such that $C$ is a $(k+1)$-cluster in $G$. W.l.o.g., let $C\subseteq A$. Consider two nodes $v\in C$ and $w\in B$. Since $k+1\leq\lambda(K)\leq deg_{K}(v)$, there is a node $x\in K$ which has an edge to $v$ in $K$.

We will show that agent $x$ can significantly improve on its strategy by reducing the capacity of the edge $(x,v)$ by 1 and creating the edge $(x,w)$ with capacity $1$ (or increasing its capacity by 1 if it already exists). We denote the obtained networks as $G'=(V',E')$ and $K'=(C,E'')$. Note that $\lambda_{G'}(x,v)\geq \lambda_{K'}(x,v)+1\geq \lambda_K(x,v)-1+1\geq k+1.$ 

If $\lambda(G')>\lambda(G)$, then agent $x$ improved on its strategy since $u_1(x,G')>u_1(x,G)$. Thus, $G$ is not a NE. 

If $\lambda(G')<\lambda(G)$, there is a cut $(X,Y), x\in X, v\in Y$ in $G'$ of size strictly less than $\lambda(G)\leq k$. On the other hand, the size of the $(X,Y)$-cut is at least $\lambda_{G'}(x,v)\geq k+1>\lambda(G)$. Hence, we have a contradiction.

Finally, we consider the case $\lambda(G')=\lambda(G)$. We will show that $u_2(x,G')>u_2(x,G)$. First, note that for any node $a\in V\setminus\{x\}$ such that $\lambda_G(x,a)>\lambda(G)$, $\lambda_G'(x,a)>\lambda(G')$ holds. Indeed, in case $\lambda_{G'}(x,a)<\lambda_G(x,a)$, we have  $\lambda_{G'}(x,a)=\lambda_{G'}(x,v)\geq k+1>\lambda(G)=\lambda(G')$. In case  $\lambda_{G'}(x,a)\geq\lambda_G(x,a)$, and since  $\lambda_G(x,a)>\lambda(G)$, we have $\lambda_{G'}(x,a)>\lambda(G)=\lambda(G')$. Therefore, after the strategy change $u_2(x,G')\geq u_2(x,G)$ holds. Second, $\lambda_G(x,w)= \lambda(G)$, since $x\in A$ and $w\in B$, whereas $\lambda_{G'}(x,w)>\lambda(G')$, i.e., $u_2(x,G')> u_2(x,G)$. Hence, $u(x,G')=(u_1(x,G),u_2(x,G'))>u(x,G)$, i.e., agent $x$ does not play a best response in $G$. This contradicts the assumption that $G$ is a NE.
\end{proof}
\noindent Now we put the pieces together to derive our tight bound.
\begin{theorem}\label{thm:PoA_mingame}
The PoA in the \mingame is $\frac{2k}{k+1}$.
\end{theorem}
\begin{proof}
By Theorem~\ref{thm:opt_produced_by_the_alg} and Theorem~\ref{thm:LB_NE_utility_mingame} we get that $\text{PoA}\leq \frac{2k}{k+1}$. By Theorem~\ref{thm:UB_NE_mingame} there is a NE network that meets the social utility lower bound, hence the PoA bound is tight.
\end{proof}

\subsection{PoA of the Avg-FlowNCG }
We start with providing constructions for two NE network with almost minimum social utility.
\begin{theorem}\label{thm:UB_NE_avggame1}
For $n$ and $k$ with $2\leq k<n$, there is a NE network  $G$ such that $u(G)=k+\frac{k(k-1)}{n-1}$ in the \avggame.
\end{theorem} 
\begin{proof}[Proofsketch]
We construct the following network $G=(V,E)$. We arrange $k$ nodes in a circle such that each node has one edge with capacity $k$ to the subsequent node. The remaining $n-k$ nodes each have a unit-weight edge to each node of the circle. See Figure~\ref{fig:NE_avggame} for an illustration of the construction. We omit the proof due to space constraints.
\end{proof}
\begin{figure}[ht]
\centering\includegraphics[width=.4\textwidth]{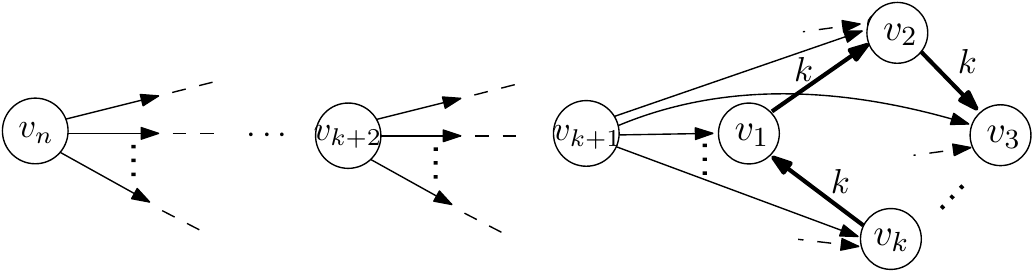}
\caption{NE network in the \avggame. Note that this network is not a NE in the \mingame since it is not $(k+1)$-edge-connected, e.g., $v_1$ can improve by reducing the capacity of the edge $(v_2,v_3)$ by 1 and increasing the capacity of the edge $(v_2,v_{k+1})$ by 1.}
\label{fig:NE_avggame}
\end{figure}

\begin{theorem}\label{thm:UB_NE_avggame2}
For a given $n$ and $k$ such that $k<n$, there is a NE network  $G$ such that $u(G)=k+1$ in the \avggame.
\end{theorem}
\begin{proof}[Proofsketch]
Consider the following star-like network $G=(V,E)$ with the set of nodes $V=\{c\}\cup\{a_1,\ldots,a_{k-1}\}\cup\{b_1,\ldots,b_{n-k}\}$. Here $c$ is a central node connected with all $a_i$ nodes by an edge with capacity 1. All $a_i$ nodes are connected with the central node by edges with capacity $k$. Each node $b_i, i=1,\ldots,n-k-1$ has an edge $(b_i,c)$ with capacity $k-1$, and only one node $b_{n-k}$ has an edge $(b_{n-k},c)$ of capacity $k$. Also, there is an edge $(c,b_1)$ with capacity 1. Finally, all $b_i$ nodes are sequentially connected by edges $(b_i,b_{i+1})$ with capacity 1. See Figure~\ref{fig:NE_avggame_2} for illustration of the construction. The proof of the statement is rather technical, so we omit it due to the space constrains.
\begin{figure}[ht]
\centering\includegraphics[width=.2\textwidth]{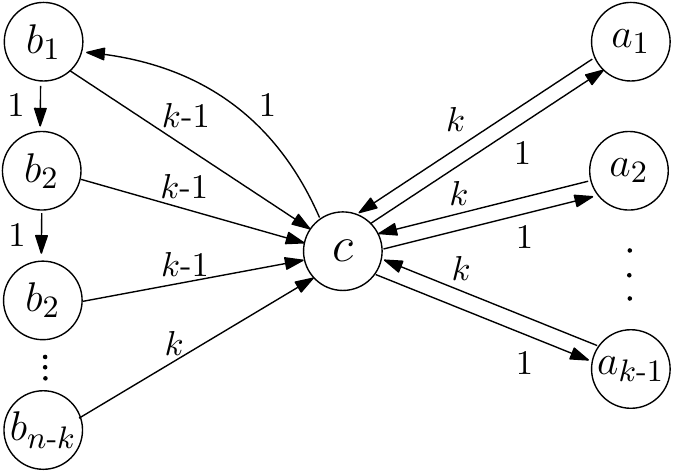}
\caption{NE network in the \avggame.}
\label{fig:NE_avggame_2}
\end{figure}
\end{proof}

\noindent A similar technique as in the proof of Theorem~\ref{thm:LB_NE_utility_mingame} can be applied to prove that the edge connectivity of any NE network in the \avggame is at least $k$. 


\begin{theorem}\label{thm:LB_NE_utility_avggame}
For $k\geq 2$, the edge connectivity of any NE network in the \avggame is at least $k$.
\end{theorem}
\begin{proof}[Proofsketch]
Assume $G$ is a NE and $\lambda(G)<k$. Consider a min-cut $(A,B)$ of size $\lambda(G)$. By performing the same procedure as in the proof of Theorem~\ref{thm:cluster} in each set $A$ and $B$, we get two sets  $X\subseteq A$ and $Y\subseteq B$ of a maximum size such that for each pair of nodes $v_1,v_2$ in the same set, $\lambda_G(v_1,v_2)\ge k$. W.l.o.g., $|Y|\geq |X|$. For some edge $(v,w)$ in a $k$-cluster in $X$, agent $v$ can reduce its capacity by 1 and create a new edge $(v,y)$ with capacity 1 to $y\in Y$ (or increase its capacity by 1 if it exists). The local edge connectivity decreases by at most 1 only between $u$ and  nodes in $X$, and increases by 1  between $u$ and all nodes in $Y$. Thus, the utility changes by at least  $-\frac{|X|-1}{n-1}+\frac{|Y|}{n-1}>0$, i.e., it is an improving move.
\end{proof}


\begin{theorem}\label{thm:PoA_UB_avggame}
For the \avggame we have PoA$<2$.
\end{theorem}
\begin{proof}
For $k=1$ a trivial NE of social utility 2 is a directed cycle with all edge capacities equal 1. 

For general $k\geq 2$, by Theorem~\ref{thm:LB_NE_utility_avggame}, any NE is at least $k$-edge-connected. On the other hand, by Theorem~\ref{thm:cluster}, any NE contains a $(k+1)$-cluster of at least two nodes. Thus, at least two nodes have their utility of at least $\frac{k+1+k(n-2)}{n-1}>k$ and all other nodes have the utility of at least $k$.  Hence the social utility is strictly grater than $\frac{k(n-2)+2k}{n}=k$. 
\end{proof}
\noindent Now we provide an almost matching lower bound on the PoA.

\begin{theorem}\label{thm:PoA_LB_avggame}
For the \avggame, we have $PoA\geq\max\Bigl\{\frac{2k}{k+k(k-1)/(n-1)}, \frac{2k}{k+1}\Bigr\}$.
\end{theorem}
\begin{proof}
By Theorem~\ref{thm:UB_NE_avggame1} and Theorem~\ref{thm:UB_NE_avggame2} the social utility of a worst-case equilibrium network is at most  $\min\{k+\frac{k(k-1)}{n-1}, k+1\}$. Therefore, for $k\leq 0.5+\sqrt{n-0.75}$ the PoA is at least $2k/\left(k+\frac{k(k-1)}{n-1}\right)=\frac{2(n-1)}{n+k-2}$. Thus, for large $n$, the PoA tends to $2$, i.e., the upper bound. With a high budget, i.e., for $k>0.5+\sqrt{n-0.75}$, the PoA is at least $\frac{2k}{k+1}$. 
\end{proof}

\section{Game Dynamics}
We analyze whether our games are guaranteed to converge to an NE under improving response dynamics. Unfortunately, this is not true for the \mingame and the \avggame.
\begin{theorem}\label{thm:IRC}
The \mingame and the \avggame do not admit the FIP.
\end{theorem}
\begin{proof}[Proofsketch]
We provide an improving response cycle (IRC) for each variant of the game depicted in Figure~\ref{fig:IRC_mingame} and Figure~\ref{fig:IRC_avggame}. Note that the first and the last networks are isomorphic in both sequences, which implies that we indeed have IRCs.
\end{proof}
\begin{figure}[ht]
\centering\includegraphics[width=.48\textwidth]{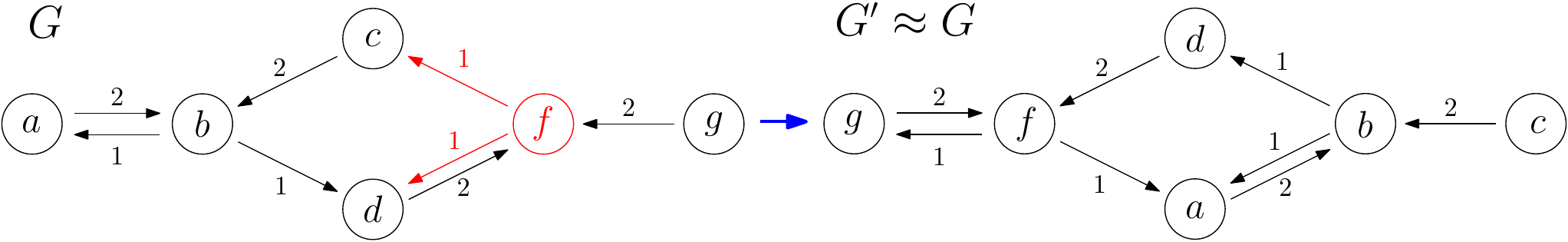}
\caption{An IRC for the \mingame with $k=2$. Since $G$ and $G'$ are isomorphic, the sequence of the improving moves recurs to~$G$.}
\label{fig:IRC_mingame}
\end{figure}

\begin{figure}[ht]
\centering\includegraphics[width=.4\textwidth]{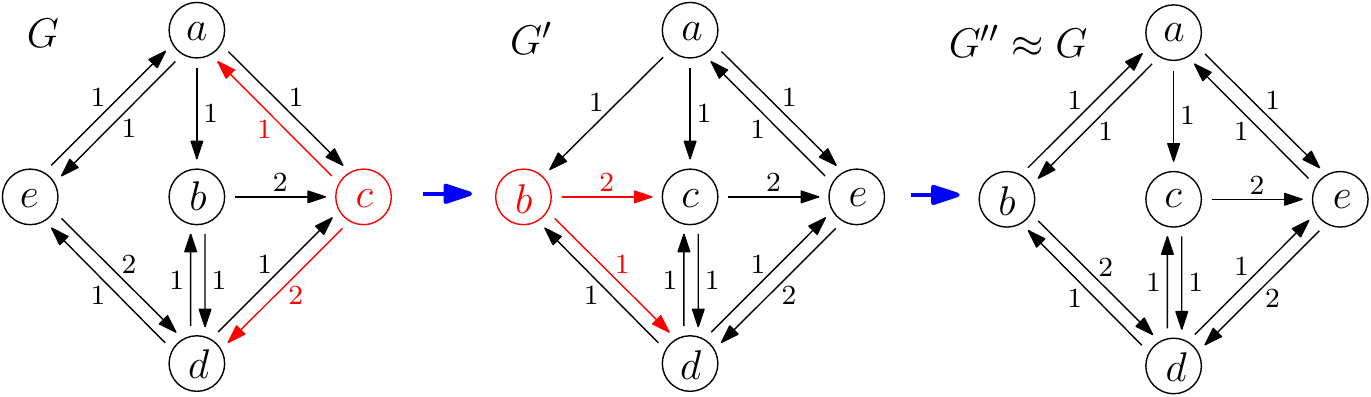}
\caption{An IRC for the \avggame with $k=3$. Since $G$ and $G''$ are isomorphic, the sequence of the improving moves recurs to~$G$.}
\label{fig:IRC_avggame}
\end{figure}

\section{Conclusion}
We made the first steps into the promising direction of investigating network formation games with bandwidth/connectivity considerations. For this, we proposed two versions: the \mingame where agents strive to maximize the connectivity of the entire network; and the \avggame where agents try to maximize their average flow value towards all other agents.
For both versions our main focus was the analysis of equilibrium networks. In particular, we proved that NE networks contain a highly-connected component in the \avggame, while in the \mingame the entire NE network must be highly-connected. This may explain the occurrence of $k$-core structures in many real-world networks. 

Our main results are (almost) tight bounds of less than $2$ on the PoA which shows that NE networks are close to optimum and no external coordination is needed in our games.

An important direction for future work is a generalization of the presented model to non-uniform budgets, non-integer edge capacities and a combination of bandwidth and centrality objectives for the agents.

\bibliographystyle{abbrv}
\bibliography{paper_bib}

\begin{thebibliography}{10}

\bibitem{ADHL13}
N.~Alon, E.~D. Demaine, M.~T. Hajiaghayi, and T.~Leighton.
\newblock Basic network creation games.
\newblock {\em SIAM Journal on Discrete Mathematics}, 27(2):656--668, 2013.

\bibitem{alvarez2019price}
C.~{\`A}lvarez and A.~Messegu{\'e}.
\newblock On the price of anarchy for high-price links.
\newblock In {\em WINE'19}, pages 316--329. Springer, 2019.

\bibitem{BG00}
V.~Bala and S.~Goyal.
\newblock A noncooperative model of network formation.
\newblock {\em Econometrica}, 68(5):1181--1229, 2000.

\bibitem{BG00_reliability}
V.~Bala and S.~Goyal.
\newblock A strategic analysis of network reliability.
\newblock {\em Review of Economic Design}, 5(3):205--228, 2000.

\bibitem{BiloL20}
D.~Bil{\`{o}} and P.~Lenzner.
\newblock On the tree conjecture for the network creation game.
\newblock {\em Theory Comput. Syst.}, 64(3):422--443, 2020.

\bibitem{CLMM16}
A.~Chauhan, P.~Lenzner, A.~Melnichenko, and M.~M{\"{u}}nn.
\newblock On selfish creation of robust networks.
\newblock In {\em {SAGT'16}}, pages 141--152, 2016.

\bibitem{chen2019network}
Y.~Chen, S.~Jabbari, M.~Kearns, S.~Khanna, and J.~Morgenstern.
\newblock Network formation under random attack and probabilistic spread.
\newblock {\em arXiv preprint arXiv:1906.00241}, 2019.

\bibitem{De07}
E.~D. Demaine, M.~T. Hajiaghayi, H.~Mahini, and M.~Zadimoghaddam.
\newblock The price of anarchy in network creation games.
\newblock {\em ACM Trans. on Algorithms ({TALG})}, 8(2):13, 2012.

\bibitem{ehsani2015bounded}
S.~Ehsani, S.~S. Fadaee, M.~Fazli, A.~Mehrabian, S.~S. Sadeghabad, M.~Safari,
  and M.~Saghafian.
\newblock A bounded budget network creation game.
\newblock {\em ACM Trans. on Algorithms (TALG)}, 11(4):34, 2015.

\bibitem{fabrikant2003network}
A.~Fabrikant, A.~Luthra, E.~Maneva, C.~H. Papadimitriou, and S.~Shenker.
\newblock On a network creation game.
\newblock In {\em PODC'03}, pages 347--351, 2003.

\bibitem{ford_fulkerson_1956}
L.~R. Ford and D.~R. Fulkerson.
\newblock Maximal flow through a network.
\newblock {\em Canadian Journal of Mathematics}, 8:399–404, 1956.

\bibitem{frank1994connectivity}
A.~Frank.
\newblock Connectivity augmentation problems in network design.
\newblock Technical report, Univ. of Michigan, Ann Arbor, MI (United States),
  1994.

\bibitem{FriedrichIKLNS17}
T.~Friedrich, S.~Ihde, C.~Ke{\ss}ler, P.~Lenzner, S.~Neubert, and D.~Schumann.
\newblock Efficient best response computation for strategic network formation
  under attack.
\newblock In {\em {SAGT'17}}, pages 199--211, 2017.

\bibitem{Goyal16}
S.~Goyal, S.~Jabbari, M.~Kearns, S.~Khanna, and J.~Morgenstern.
\newblock Strategic network formation with attack and immunization.
\newblock In {\em WINE'16}, pages 429--443. Springer, 2016.

\bibitem{JW96}
M.~O. Jackson and A.~Wolinsky.
\newblock A strategic model of social and economic networks.
\newblock {\em Journal of economic theory}, 71(1):44--74, 1996.

\bibitem{KL13}
B.~Kawald and P.~Lenzner.
\newblock On dynamics in selfish network creation.
\newblock In {\em {SPAA'13}}, pages 83--92. ACM, 2013.

\bibitem{Kli11}
L.~Kliemann.
\newblock The price of anarchy for network formation in an adversary model.
\newblock {\em Games}, 2(3):302--332, 2011.

\bibitem{kliemann2017swap}
L.~Kliemann, E.~Shirazi~Sheykhdarabadi, and A.~Srivastav.
\newblock Swap equilibria under link and vertex destruction.
\newblock {\em Games}, 8(1):14, 2017.

\bibitem{koutsoupias1999worst}
E.~Koutsoupias and C.~Papadimitriou.
\newblock Worst-case equilibria.
\newblock In {\em STACS'99}, pages 404--413. Springer, 1999.

\bibitem{laoutaris2008bounded}
N.~Laoutaris, L.~J. Poplawski, R.~Rajaraman, R.~Sundaram, and S.-H. Teng.
\newblock Bounded budget connection ({BBC}) games or how to make friends and
  influence people, on a budget.
\newblock In {\em PODC'08}, pages 165--174. ACM.

\bibitem{MMO15}
E.~A. Meirom, S.~Mannor, and A.~Orda.
\newblock Formation games of reliable networks.
\newblock In {\em INFOCOM '15}, pages 1760--1768. IEEE, 2015.

\bibitem{MS12}
M.~Mihal\'{a}k and J.~C. Schlegel.
\newblock Asymmetric swap-equilibrium: A unifying equilibrium concept for
  network creation games.
\newblock In {\em MFCS'12}, pages 693--704. 2012.

\bibitem{monderer1996potential}
D.~Monderer and L.~S. Shapley.
\newblock Potential games.
\newblock {\em Games and economic behavior}, 14(1):124--143, 1996.

\bibitem{Myerson13}
R.~B. Myerson.
\newblock {\em Game theory}.
\newblock Harvard university press, 2013.

\bibitem{nagamochi2008algorithmic}
H.~Nagamochi and T.~Ibaraki.
\newblock {\em Algorithmic aspects of graph connectivity}.
\newblock Cambridge University Press, 2008.

\bibitem{Pap01}
C.~H. Papadimitriou.
\newblock Algorithms, games, and the internet.
\newblock In {\em STOC'01}, pages 749--753, 2001.

\bibitem{T04}
v.~Tardos.
\newblock Network games.
\newblock In {\em STOC'04}, pages 341--342, 2004.

\bibitem{watanabe1987edge}
T.~Watanabe and A.~Nakamura.
\newblock Edge-connectivity augmentation problems.
\newblock {\em Journal of Computer and System Sciences}, 35(1):96--144, 1987.

\end{thebibliography}
\end{document}